\documentclass{vldb}

\pdfoutput=1

\usepackage{tikz}
\usepackage{xspace}
\usepackage{balance}
\usetikzlibrary{positioning}
\usepackage{booktabs}

\newtheorem{theorem}{Theorem}
\newtheorem{definition}{Definition}
\newtheorem{lemma}[theorem]{Lemma}
\newtheorem{corollary}[theorem]{Corollary}
\usepackage[vlined,boxed]{algorithm2e}

\DeclareMathOperator*{\argmin}{arg\,min} 
\newcommand{\dom}[2]{#1 \,\prec_\text{dom}\, #2} 
\newcommand{\arc}{edge\xspace} 
\newcommand{\arcs}{edges\xspace}

\newcommand{\vertices}{\mathcal{V}}
\newcommand{\edges}{\mathcal{E}}
\newcommand{\costs}{\mathcal{C}}

\newcommand{\vectors}[1]{\mathbb{R}^{#1}_+}

\newcommand{\paths}{\mathcal{S}}
\newcommand{\indices}[1]{\mathcal{I}}

\newcommand{\locally}[1]{#1}
\newcommand{\globally}[1]{globally #1}

\newcommand{\gglobal}[1]{global #1}

\newcommand{\Global}[1]{Global #1}

\newcommand{\nodeskyline}[1]{\mathcal{S}(#1)}

\newcommand{\succname}{succ}
\newcommand{\succiname}[1]{\succname{}_i}
\newcommand{\succi}[1]{\succiname{}(#1)}

\newcommand{\pathcost}[1]{cost(#1)}
\newcommand{\edgecost}[2]{cost(#1,#2)}

\newcommand{\open}{open}

\newcommand{\nondominated}{skyline} 
\newcommand{\lbname}{lb}
\newcommand{\lb}[1]{\lbname{}(#1,t)}

\newcommand{\lbbetaname}{{\lbname{}}^{\beta}}
\newcommand{\lbbeta}[1]{\lbbetaname{}(#1)}
\newcommand{\lbcost}[2]{lb(#1,#2)}

\renewcommand{\time}{dur}
\newcommand{\distance}{len}
\newcommand{\energy}{ener}
\newcommand{\crossings}{cros}
\newcommand{\timelights}{{dur}$_{\text{p}}$}

\newcommand{\crsep}{/}

\newcommand{\criteria}[1]{{\scriptsize #1}}

\newcommand{\captiontext}[2]{#1 of #2.}
\newcommand{\comparedalgorithmsbi}{previous approach DD and proposed approaches PP and BPP}
\newcommand{\comparedalgorithmsmulti}{previous approach MD and proposed approaches PP and BPP}

\newcommand{\compareruntime}{Runtime comparison}
\newcommand{\comparenodes}{Search area comparison}

\newcommand{\munichtasks}{702 Munich routing tasks}
\newcommand{\bavariatasks}{90 Bavaria routing tasks}

\newcommand{\td}{\criteria{\time{}\crsep{}\distance{}}}
\newcommand{\te}{\criteria{\time{}\crsep{}\energy{}}}
\newcommand{\tc}{\criteria{\time{}\crsep{}\crossings{}}}
\newcommand{\tl}{\criteria{\time{}\crsep{}\timelights{}}}

\newcommand{\tel}{\criteria{\time{}\crsep{}\energy{}\crsep{}\timelights{}}}
\newcommand{\tce}{\criteria{\time{}\crsep{}\energy{}\crsep{}\crossings{}}}
\newcommand{\tdc}{\criteria{\time{}\crsep{}\distance{}\crsep{}\crossings{}}}

\newcommand{\tdecl}{\criteria{all 5 criteria}}

\newcommand{\vecmax}[2]{\text{max}\,(#1,\, #2)}

\newcommand{\LC}{LCS }
\newcommand{\LCSS}{$\text{LCS}_\text{ss}$ }

\newcommand{\REFLC}[1]{REF$_{#1}$+LCS}

\newcommand{\MDLC}{MD+LCS}
\newcommand{\PPLC}{PP+LCS}

\newcommand{\solgr}{G_\text{SOL}}
\newcommand{\ppgr}{G_\text{PP}}
\newcommand{\soledges}{\edges_\text{SOL}}
\newcommand{\ppedges}{\edges_\text{PP}}

\title{ParetoPrep: Fast computation of Path Skylines Queries}

\author{
Michael Shekelyan \quad
Gregor Joss\'e \quad
Matthias Schubert
\\
 \affaddr{Institute for Informatics, Ludwig-Maximilians-Universit\"at M\"unchen}\\
       \affaddr{Oettingenstr.\ 67}\\
       \affaddr{D-80538 Munich, Germany}\\
       \affaddr{\{shekelyan,josse,schubert\}@dbs.ifi.lmu.de}
}
\begin{document}

\maketitle

\begin{abstract}
Computing cost optimal paths in network data is a very important task in many
application areas like transportation networks, computer networks or social
graphs. In many cases, the cost of an edge can be described by various cost 
criteria. For example, in a road network possible cost criteria are distance,
time, ascent, energy consumption or toll fees. In such a multicriteria network,
a route or path skyline query computes the set of all paths having pareto
optimal costs, i.e. each result path is optimal for different user preferences.
In this paper, we propose a new method for computing route skylines which
significantly decreases processing time and memory consumption. Furthermore, our
method does not rely on any precomputation or indexing method and thus, it is
suitable for dynamically changing edge costs. Our experiments
demonstrate that our method outperforms state of the art approaches
and allows highly efficient path skyline computation without any preprocessing.
\end{abstract}

\section{Introduction}
In recent years, querying network data has become more and more important in
many application areas like transportation systems, the world wide web, computer networks
or social graphs.
One of the most important tasks in network data is computing cost optimal
paths between two nodes. Especially in transportation networks and
computer networks finding shortest or cost optimal paths is essential for
optimizing the movement of objects or information. An optimal path can depend on
multiple cost criteria and a network considering different cost dimensions for
each link is called a multicost or multicriteria network. 
In road networks, possible optimization criteria are travel time,
travel distance, toll fees, environmental hazards or energy consumption.
In computer networks, typical cost criteria are bandwidth, rental cost and
current traffic.

A simple solution to compute optimal paths in multicriteria networks
is to combine all cost values into a single optimization criterion. However,
finding a suitable combination method is a difficult task because a good choice
might depend on personal preferences and the current query context.

An alternative approach is to compute all pareto optimal paths,
i.e. all paths that are potentially cost optimal under any cost
combination. In the database community, this task is generally known as skyline
query \cite{BorKosSto01} and when looking for routes in a network as route
skyline query \cite{KriRenSch10}. In other communities like Artificial Intelligence and
Operations Research the latter task is known as multicriteria shortest path
optimization \cite{classification}.

One of the most efficient approaches to answering route skyline queries is the
label correcting approach which exploits the pareto optimality of any subpath
of a result path. However, for large graph this pruning mechanismen is not
enough to ensure efficient computation because the majority of the network has
to be visited. To efficiently compute path skylines in large networks, it is
necessary to direct the search towards the target, e.g. by using lower bound
approximations for the remaining cost in each criteria in a similar way as
A*-search does in the single criterion case.

For example, a lower bound approximation for network distance is the Euclidean distance.
Considering multiple and arbitrary cost criteria requires a more general approach for acquiring lower
bound approximations. One solution is precalculating bounds like a Reference
Node Embedding \cite{GolHar04}. However, precalculated bounds have several drawbacks.
Firstly, the approximation quality is often insufficient for the majority of queries.
Secondly, they typically require a lot of memory compared to the size of
the network. And finally, in dynamic networks -- where edge costs vary over time
-- the approximation may lose its bounding properties.

Alternatively, \cite{MultiDijkstra} makes use of the following skyline property:
For every particular cost criterion there exists an optimal path and each 
nondominated path cannot have a smaller cost value for this criterion.
Now, by performing a single-source all-target Dijkstra search, it is possible
to compute the single criterion optimal paths from each node in the network
to the target node and thus, derive lower bounds for each cost value w.r.t. the
considered cost criterion. Performing this search for all cost
criteria, computes a vector of lower bounds for each node. 
Although the effort is large, the quality of the bounds compensates for the overhead
when computing the path skyline later on. However, the runtime degenerates for large
graphs and multiple criteria because the entire graph is visited once for
each criterion.

To avoid the drawbacks of both approaches, we propose a new method,
\emph{ParetoPrep}, which strongly increases the efficiency even for multiple
cost criteria. Our method does not rely on preprocessing, hence it is applicable
to networks with dynamically changing costs. 
ParetoPrep computes tight lower bound approximations which yield great pruning power
and low runtime. Additionally, ParetoPrep finds the shortest paths between a start and a target node
with respect to all $d$ cost criteria within a single graph traversal.
When traversing the graph ParetoPrep visits only a small part of the network,
keeping the search local. To even further reduce the explored part of the graph,
we introduce a bidirectional optimization of ParetoPrep.

In our experimental evaluation, we show that our new approach outperforms the state of the art
in computing route skylines, i.e. the above mentioned approach of computing lower bounds
through multiple Dijkstra searches \cite{MultiDijkstra} as well as the ARSC \cite{KriRenSch10}. 
As test network, we use the open street map road network of the German
federal state of Bavaria and examine searches for up to five different cost criteria. 
To conclude, the contributions of this paper are:
\begin{itemize}

  \item We introduce an algorithm, ParetoPrep, which computes tighter lower bound cost approximations than any state of the art approach.
  \item We show that ParetoPrep outperforms any comparable algorithm in terms of runtime, memory space and effectiveness.
	\item We present a complete algorithm for answering path skyline queries on large
			and possibly dynamic network data and a bidirectional optimization of
			this algorithm.
\end{itemize}

The rest of the paper is organized as follows:
In Section~\ref{section:relatedwork} we present related work in the area of
path skyline computation and routing in large networks. Section~\ref{section:preliminaries}
 provides basic notations, concepts and
formalizes path skyline queries. ParetoPrep, our new algorithm,
is presented in Section~\ref{section:contribution}. Additionally, the section
contains formal proofs concerning the correctness of ParetoPrep.
Section~\ref{section:experiments} describes the results of our experiments,
demonstrating the various benefits of ParetoPrep.
The paper concludes with a summary and directions for future
work in Section~\ref{section:summary}.
\section{Related Work}\label{section:relatedwork}
The task of finding all pareto optimal elements within a vector space
was introduced to the database community as skyline operator in \cite{BorKosSto01}.
Since then, multiple methods for computing
skylines in database systems on sets of (cost) vectors
followed \cite{TanEngOoi01,KosRamRos02,PapTaoFuSee03}. 

With regard to network data and especially spatial networks there exist some approaches 
for computing skylines on landmarks or other points of interest. 
In \cite{DenZhoShe07} the authors introduce a method for calculating a skyline of
landmarks in a road network which are compared w.r.t. their network distance to
several query objects. In \cite{HuaJen04} the authors propose in-route skyline
processing in road networks. Assuming that a user is moving along a predefined
route to a known destination, the algorithm processes minimal
detours to sets of landmarks being distributed along the path. In
\cite{JanYoo08} the authors discuss continuous skyline queries in
road networks, i.e. a moving user queries for a skyline of points of interest.

The task of computing a skyline of pareto optimal paths between two nodes is 
referred to as route skyline query in \cite{KriRenSch10}. However, in accordance
the graph theoretical terminology, we refer to it as path skyline query.
In the following, we will discuss solutions from other areas of research and discuss the differences between the
algorithm ARSC proposed in \cite{KriRenSch10} and alternative approaches.

In Operations Research, the problem of finding complete path skylines is known 
as the Multiobjective Shortest Path problem and surveys on existing solutions 
to this problem can be found in \cite{classification,ehrgott, tarapata, ehrgott2000survey}.
Early on, \cite{hansen} proved that the size of the path skyline may
increase exponentially with the number of hops between start and target node,
and that the problem therefore is NP-hard. More recently, \cite{feasible} showed 
that the number of routes is in practice feasibly low when using strongly correlated cost criteria.

The potential complexity of the problem motivated solutions which approximate
the pareto front. 
A fully
polynomial-time approximation scheme is provided by \cite{hansen}, \cite{fptas} and \cite{fptas2}.
Many approximative solutions are also based on genetic algorithms \cite{surveygenetic}.

Concerning exact solutions, the most common approach are
labeling algorithms which label nodes with the cost vectors of assembled paths ending at that node.
They start with paths consisting of the outgoing edges of the start node,
and in each iteration extend all previously
assembled paths that may be part  of a skyline path and terminate once all
assembled paths were either extended or pruned. Paths can be pruned if they are dominated
by a path with the same terminal node because all subpaths of nondominated
paths are also nondominated. We will refer to this type of pruning as \emph{local domination}.
To perform local domination checks, a route skyline is maintained for each node, which
consists of all so far nondominated routes ending at that node.

Labeling approaches can be subdivided into label setting and label correcting approaches.
In each iteration, label setting approaches like Martin's algorithm \cite{martins} extend 
a nondominated path and therefore have a well-defined worst case performance, limited
by the number of nondominated paths. In contrast, label correcting approaches do not 
necessarily extend nondominated paths which offers more flexibility. For
instance, the label correcting approach allows to extend the complete route
skyline of a node in each iteration, instead of separately extending particular
paths.

In the special case of two cost criteria, node skylines can be sorted to improve the
efficiency of local domination checks. Two dimensional skyline cost vectors sorted 
in ascending order by the first criterion are simultaneously sorted in descending order  
w.r.t. to the second criterion. When the route skyline of node $u$ is extended with 
edge $(u,v)$, it has to be merged with the set of nondominated paths ending at $v$. 
\cite{brumbaugh} shows how to efficiently solve the problem of merging two sorted route skylines and
\cite{skriver2000} describes an efficient method to check if one skyline
completely dominates the other making the merging step trivial.

The labeling algorithms cited so far solve the single-source problem of
finding all nondominated paths from $s$ to all other nodes. In \cite{moa},
it is shown that when solving the point-to-point problem
from a start node $s$ to  a target node $t$ all partially or fully extended
paths dominated by a path from $s$ to $t$ can be pruned. We will refer to this
type of pruning as global domination.
Global domination allows to process only a part of the graph and significantly
reduces the number of considered paths to compute the path skyline. Thus, to
answer path skyline queries in large graphs \gglobal{domination} should be employed

to direct the search. \cite{moa} also showed that lower bound cost estimations of fully extended paths
can be utilized to considerably improve the pruning power of \gglobal{domination} and can be used to prioritize the extension of more promising paths.

For cost criteria other than distance, lower bound cost estimations are 
usually acquired through some kind of preprocessing.
In ARSC \cite{KriRenSch10}, lower bounds are provided by a Reference Node
Embedding computed before any query is posed. However, a
reference point embedding usually does not yield good lower bounds for
all possible queries. Furthermore, the memory consumption is rather large and
any precomputed information has to be checked for validity in case of dynamically
changing edge costs. Another approach is to compute lower bound costs 
individually for each query \cite{machuca2012}.
Such lower bound cost computation may be combined with
computing the shortest paths for each criterion which can then be used for
\gglobal{domination} checks. Tung and Chew \cite{MultiDijkstra} proposed to
perform a single-source all-target Dijkstra search for each cost criterion on
the graph in reverse direction \cite{dijkstra} to find the
costs of the shortest paths from all other nodes in the graph to the target node $t$. We will refer to this
approach as Multidijkstra (MD). A major shortcoming of this approach is that it 
processes the whole graph for every cost criterion.

In the special case of two cost criteria, Machuca and Mandow
\cite{doubledijkstra} propose a termination condition which makes use of
special properties of two dimensional skylines. We would like to stress that
this method is not applicable to the case of more than two cost criteria. Since the algorithm in
\cite{doubledijkstra} uses two Dijkstra searches, we shall refer to it as 
Doubledijkstra (DD).
In \cite{YanGuoJenetAl14} the authors compute a set of optimal paths in
time-dependent, uncertain, multicriteria networks. The presented model
allows the complex representation of a dynamic network model. However, the lower bound
computation relies on the same undirected search as in \cite{MultiDijkstra}.

It should be stressed that MD and DD do not perform point-to-point shortest path searches.
Therefore, speed-up techniques for point-to-point searches are not applicable.

The methods presented in this paper use significantly less runtime and
memory space compared to the MD and DD approaches by restricting the visited part of the
network and deriving the required lower bounds in a single graph traversal.
Thus, Pareto Prep can be combined with any search algorithm requiring time and memory
efficient computation of lower bound costs for multiple attributes.

\section{Preliminaries}
\label{section:preliminaries}
\subsection{Route Skyline Queries}

A multicost network is represented by a directed weighted graph
$G(\vertices,\edges,\costs)$ comprised of a set of nodes $\vertices$ and a set
of directed \arcs{} $\edges \in \vertices \times \vertices$. Each \arc $(n,m) \in \edges$ is
labeled with a cost vector $\edgecost{n}{m} \in \costs \subset \vectors{d}$
which consists of the costs for traversing \arc  $(n,m)$ w.r.t. each of the $d$ cost criteria
(in mathematical context also referred to as cost dimensions). 
If there exists an \arc $(n,m)$ then $n$ and $m$ are neighboring nodes and $(n,m)$ shall be 
called an outgoing \arc of $n$ and an incoming \arc of $m$. 

A sequence of adjacent edges connecting
two nodes $s$ and $t$, $w=((s,n_1),(n_1,n_2),\ldots,(n_k,t))$ , 
is called a way from $s$ to $t$. If $w$ does not visit any
node twice it is called a path or a route.
The cost of a path $p$ in the $i$-the cost dimension is the sum of its individual \arc costs:
$$\pathcost{p}_i = \sum_{(n,m) \in p} \edgecost{n}{m}_i $$

A cost vector $a \in \vectors{d}$ dominates a cost vector $b \in
\vectors{d}$, denoted $\dom{a}{b}$, iff $a$ has a smaller cost value than $b$ in at least one dimension $i$ and $b$ does
not have a smaller cost value than $a$ in any dimension $j$:
$$ \exists_{i \in \{1,\ldots,d\}} : a_i < b_i \ \land \ {\nexists}_{j \in \{1,\ldots,d\}} : a_j > b_j \ :\Leftrightarrow \ \dom{a}{b} $$

For a set of paths $A$, $\dom{A}{b}$ implies that there exists some path $a\in A$ which dominates the path $b$. 
Cost vectors which are not dominated by any other cost vector are called nondominated or 
pareto optimal. The set of nondominated cost vectors
includes all optimal trade-offs between cost criteria and is called the pareto front
or skyline \cite{BorKosSto01}.

\begin{definition}[Path Skyline]
Let $G(\vertices,\edges,\costs)$ be a multicost network, let $s \in
\vertices$ be the start node and let $t \in \vertices$ be the target node. Then,
the set of all paths between $s$ and $t$ whose cost vectors are nondominated is
defined as path skyline $\mathcal{S}(s,t)$.
\end{definition}

The task we examine in this paper is to efficiently compute
$\mathcal{S}(s,t)$. In the following, we always assume a multicost network is given.

\subsection{Label Correction and Pruning}
Our solution employs the label correcting approach and incorporates two domination 
checks in order to prune paths from the search tree. The first type of
domination is local domination which is used to compare paths starting and
ending at the same node.
The second domination relationship is global domination
it describes whether a path still can be extended to a nondominated path between
the start and the target node of the given query.

\begin{definition}[Local Domination]
Let $p$ and $q$ be two paths starting at the same node $s$ and ending at the
same node $t$. Iff $\dom{\pathcost{p}}{\pathcost{q}}$ holds, we refer to $q$
as \locally{dominated} by $p$ : $\dom{p}{q}$.
Correspondingly, $q$ is referred to as \locally{nondominated} iff $\nexists
\ p=((s,n_1),\ldots,(n_k,t))\,: \dom{\pathcost{p}}{\pathcost{q}}$.
\end{definition}

For any path $q=((n_1,n_2),\ldots,(n_{l-1},n_l))$, any subsequence of 
edges $p=((n_i,n_{i+1}),\ldots,(n_{k-1},n_k))$ is called a subpath of $q$.
The following lemma shows that locally dominated paths can be pruned from the
search.

\begin{lemma}[Local Domination Check]\label{lem:localdomination}
Any subpath of a \locally{nondominated} path is \locally{nondominated} (w.r.t. its start and target node). 
\end{lemma}

\begin{proof}

Suppose, a path $p$ had a \locally{dominated} subpath $q$. This subpath could be
substituted with the \locally{dominating} path $q'$. The modified path $p'$ would then
dominate the original path $p$:
$$ cost(p') = cost(p) - cost(q) + cost(q') $$
$$ \dom{cost(q')}{cost(q)} \Rightarrow \dom{cost(p')}{cost(p)} $$
This contradicts the assumption and thereby proves the claim.

\end{proof}

Using local domination allows to decide which paths between node $s$ and some
other node $n$ have to be extended when processing node $n$. However, it is
sufficint to decide whether any path between $s$ and $n$ can be extended into a
nondominated path between $s$ and $t$. To test this characteristic, we 
introduce global domination.

\begin{definition}[Global Domination]
Given a start \\ node $s$ and a destination node $t$, an arbitrary path $p$ is called 
\globally{dominated} iff it is a subpath of a \locally{dominated} path
$q$ between $s$ and $t$. Respectively, any subpath $p$ of a
\locally{nondominated} path $q$ between $s$ and $t$ is called
\globally{nondominated}.
\end{definition}

Due to global domination it is not necessary to extend any path $p$ between $s$
and some other node $n$, if there does not exist any extension $q$ of $p$
ending an the target $t$ and being part of $\mathcal{S}(s,t)$. Without the

concept of global domination an algorithm cannot stop the search until all
path skylines $\mathcal{S}(s,v)$ for all $v \in \vertices$ are calculated.
To exploit the fact that we are only interested in $\mathcal{S}(s,t)$, a simple
global domination check is to compare the cost of any path $p$ to 
$\mathcal{S}(s,t)$. If $\exists q \in \mathcal{S}(s,t): \dom{cost(q)}{cost(p)}$
then $p$ cannot be extended into an element of $\mathcal{S}(s,t)$ because
$cost(p)$ cannot be decreased by extending $p$ for  any cost criterion $i$.
This basic globale domination check is equivalent to the stopping criterion of
Dijkstra's algorithm for single criterion shortest path computation.
However, the check is not sufficient to restrict
the search space for the path skyline to a limited portion of the network. 
To check for global domination in a more effective way, lower bound cost
approximations are employed.
Before explaining the global domination check being employed in this paper, we
need to introduce some notation:\\
$\vecmax{a}{b}$ shall denote the component-wise (criterion-wise) maximum of
cost vectors $a$ and $b$, i.e. $(\vecmax{a}{b})_{1\leq i\leq d} := \max\{a_i,b_i\}$.

\begin{definition}[Lower Bound Costs]
Let $n$ and $m$ be nodes and let $i$ be an arbitrary cost
criterion. If for the real value $\text{lb}(a,b)_i$ holds $\text{lb}(a,b)_i \leq
cost(p)_i$ for any path $p$ connecting $a$ and $b$ then $\text{lb}(a,b)_i$ is called lower
bound for $i$ w.r.t. $a$ and $b$. The vector consisting of the lower bounds of all cost
criteria w.r.t. $a$ and $b$ is denoted as $\lbcost{a}{b}$.
\end{definition}

To compute lower bounds, it is possible to employ the MultiDijkstra (MD)
method which is computing all single criterion optimal paths starting at some
node $n$ and ending at node $t$ as described in section
\ref{section:relatedwork}. In the next section, we will introduce our new method
ParetoPrep which is computing even tighter lower bounds in much more efficient
way. Additionally, to computing lower bound cost both algorithms compute the
optimal paths between the start node $s$ and the target node $t$ for each cost
criterion $i$. Both information are employed to prune paths from the search
space in the following domination check:

\begin{lemma}[Global Domination Check] \label{lemma:globaldominationcheck}
Let $p$ be a path from $n$ to $m$ and $q$ be a path from the start node $s$ to
the target node $t$. If
$$\dom{\pathcost{q}}{\vecmax{\lbcost{s}{n}+\pathcost{p}+\lbcost{m}{t}}{\lbcost{s}{t}}}$$,
then the path $p$ is \globally{dominated}.
\end{lemma}

\begin{proof}
The cost vector $\lbcost{s}{n}+cost(p)+\lbcost{m}{t}$ is the lower bound cost of all paths from $s$ to $t$
via $p$. If this lower bound cost is dominated by the cost of a path from $s$ to $t$, there does not
exist a \locally{nondominated} path from $s$ to $t$ via $p$. Additionally, there
cannot be any path $q$ from $s$ to $t$ where $cost(q)_i$ is smaller than the
cost of the shortest path w.r.t. to criterion $i$. Thus, the bound can always be
increased to the cost of the shortest path in dimension $i$.
\end{proof}

Now that we have introduced our pruning criteria, let us explain the label
correcting algorithms being used in the paper:
Any label correcting approach employs two important data structures. The first is 
a node table maintaining an entry for each visited node. The entry of
node $n$ maintains a local skyline, i.e. a list of nondominated paths starting
in $s$ and ending in $n$. For each path a flag is maintained, denoting
whether the path has already been processed or not. The second data structure
is a queue of node entries which is ordered with respect to a preference
criterion. In our implementation, we use the sum over the minimal cost values of the
local skyline in each cost criterion, but other preference function work as
well.

When answering a path skyline query for a given start node $s$ and a target
node $t$, the search starts by visiting the outgoing edges of $s$.
For each newly reached node, a node tab entry is generated and added to the
queue. In the main loop of the algorithm the top entry
is removed from the queue. Any path in its local skyline which has neither been processed
nor is \globally{dominated} by any result path found so far is extended. That means,
the path is extended by all outgoing links of node $n$. A new path is registered
in the local skyline of the respective target node if it is not \locally{dominated} by the
other paths in that entry. If a node tab entry is updated with a new path
it is added to the queue. The algorithm terminates when the queue is empty, i.e. no entry
contains any unprocessed path. This algorithm is similar to the
ARSC algorithm introduced in \cite{KriRenSch10}. However, the global domination
check in ARSC is less strong because it cannot exploit the existence of the
single criterion shortest paths. Furthermore, the lower bounds are less tight due to the
use of the precomputed Reference Node Embedding.

\begin{algorithm}[t]
\fontsize{8}{10}
\SetKwInOut{Input}{input}
\SetKwInOut{Output}{output}
\DontPrintSemicolon
\BlankLine
\tcp{step 1: initialisation}
\BlankLine
\ForEach{outgoing edge $(s,m)$ of $s$}{
	\BlankLine
	$\nodeskyline{m} \leftarrow \nodeskyline{m} \ \cup$ path consisting of the edge $(s,m)$ \;
	$\open \leftarrow open \ \cup \ \{m\}$ \;
}
\BlankLine
\tcp{step 2-4: main loop}
\BlankLine
\While{$\open{} \neq \{ \}$}{
	\tcp{step 2: node selection }
	$n \leftarrow {\scriptstyle \argmin_{o \in \open} \left( \min_{p \in \nodeskyline{o}} \sum_{i = 1}^{d} cost(p)_i+\lbcost{o}{t}_i \right)} $\;
	$\open \leftarrow open \setminus \{n\}$ \;
	\BlankLine
	
	\tcp{step 3: mark \globally{dominated} as processed }
	
	\BlankLine
	\ForEach{unprocessed $p \in \nodeskyline{n}$}{		
		\If{$\dom{\nodeskyline{t}}{\vecmax{\lbcost{s}{t}}{cost(p)+\lbcost{n}{t}}} $}{
			\BlankLine
			mark $p$ as processed\;
		} 
	}
	
	\BlankLine
	
	\tcp{step 4: extend unprocessed paths in $\nodeskyline{n}$}
	
	\BlankLine
	
	$\mathcal{A} \leftarrow$ unprocessed paths in $\nodeskyline{n}$ \;
	\BlankLine
	\ForEach{outgoing edge $(n,m)$ of $n$}{
		
		
		$\mathcal{B} \leftarrow$ paths in $\mathcal{A}$ extended with $(n,m)$ \;
		$\nodeskyline{m} \leftarrow $ merged skyline of $\nodeskyline{m}$ and $\mathcal{B}$ \;
		
		\BlankLine
		\If{ $\nodeskyline{m}$ was modified}{
		\BlankLine				
			$\open \leftarrow open \cup \{m\}$ \;
		}
	}
	
	mark all paths in $\nodeskyline{n}$ as processed \; 
}
\BlankLine
\tcp{step 5: termination}
\BlankLine
\Return{$\nodeskyline{t}$}
\vspace{0.5eM}
\caption{Pseudocode of ARSC, a label-correcting search for all nondominated paths from $s$ to $t$.} 
\label{alg:arsc}
\end{algorithm}

\subsection{Special Case: Two Cost Criteria}
\label{sec:twocriteria}

Although we focus on the general multicriteria case, we want to point out there
exist two major optimizations for the special case of two cost criteria. 
One pertains to the problem of skyline merging, the other significantly limits the
number of nodes that have to be visited. Both optimizations stem from the same 
quality of two-dimensional skylines: Sorting cost vectors w.r.t. one criterion
implies a sorting in reversed order w.r.t. the other criterion.

This can be used to more efficiently maintain local skylines in sorted lists
\cite{brumbaugh}. As a result, the only paths potentially dominating a path $p$ are

its two neighbors in the sorted skyline. Furthermore, the set of potentially dominated
skyline paths is also restricted to neighboring paths \cite{SheJosSchKri14}.
We employ this optimization for all compared algorithms in the experiments using only two
cost criteria.

The above also implies that the optimal path w.r.t. one criterion
(by definition part of the skyline) is the weakest skyline path w.r.t. to the other criterion.
Hence, if $s^1$ is the optimal path from $s$ to $t$ w.r.t. the first criterion, then $\pathcost{s^1}_2$
constitutes the upper bound for the second criterion, and vice versa. 

This property can be used to restrict the search space of lower bound
computation of the MD approach for the two criteria case and we will refer to
this approach as Doubledijkstra (DD) \cite{doubledijkstra}. The optimization
starts with a reverse dijkstra search for the shortest path from $t$ to $s$ w.r.t. the first
criterion. After this search reaches $s$, the search is done w.r.t. the second
criterion. Thus, we have upper bounds for both criteria. Finally, the search
for both criteria is continued until all nodes having one cost value
being smaller than the corresponding upper bound are found. Let us note that
this optimization is only applicable to the MD approach and it is still significantly
less efficient than our new method ParetoPrep due to lower quality of the lower bounds.

\section{Lower Bound Computation}
\label{section:contribution}

In this section, we introduce our novel approach ParetoPrep which computes all
required lower bounds for efficiently computing the path skyline
$\mathcal{S}(s,t)$. The idea of ParetoPrep is to compute all
single criterion shortest paths between $s$ and $t$ within a single partial graph
traversal. It will be shown that there cannot be a 
nondominated path between $s$ and $t$ containing any edge that is not visited
during this traversal. Thus, ParetoPrep visits all required nodes and
edges for processing a path skyline query. 
At the end of the section, we will introduce bidirectional
ParetoPrep which further limits the nodes being visited during the lower
bound computation.

\begin{figure}[t]
\centering
\includegraphics[height=5cm]{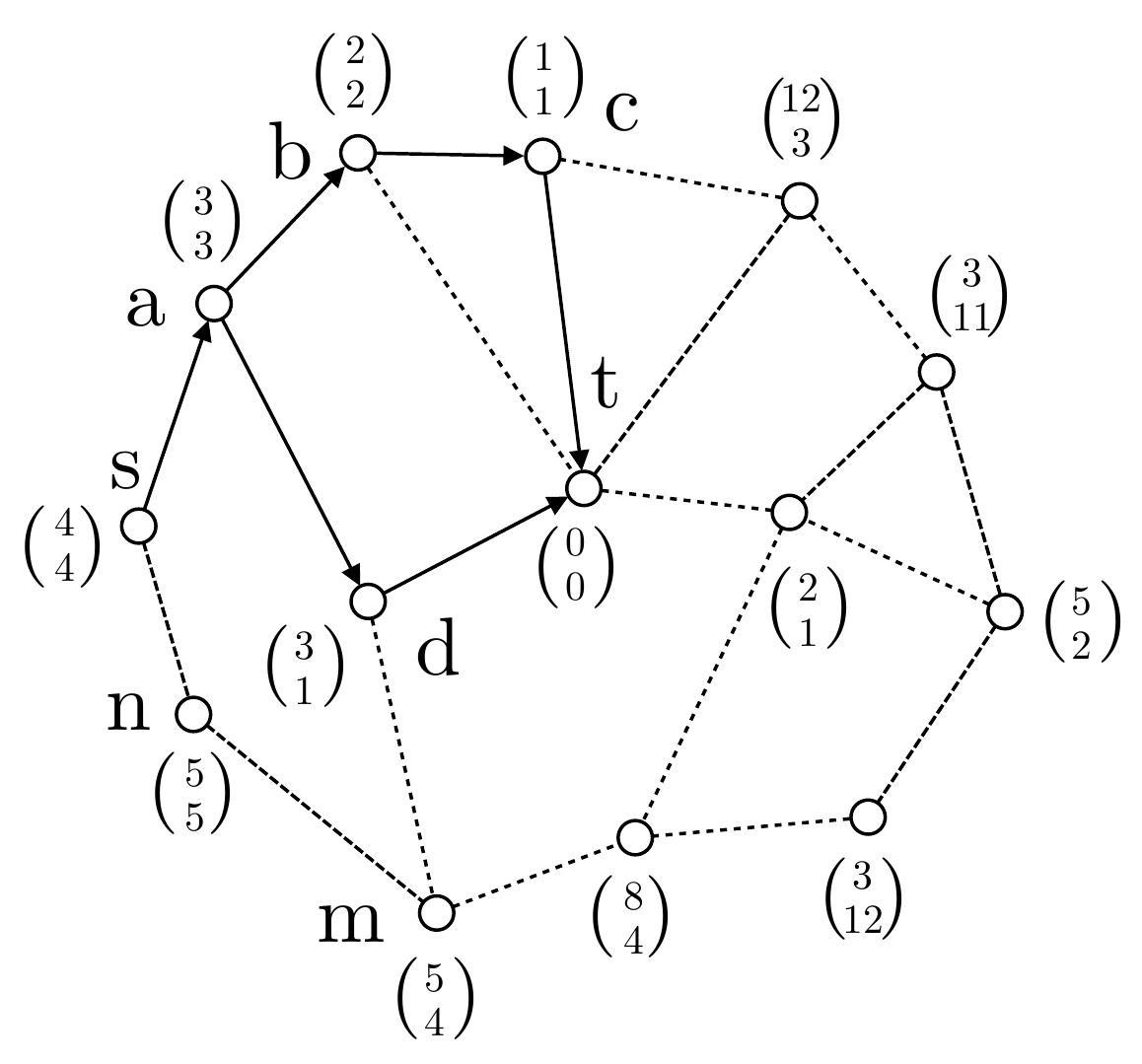}
\caption{Exemplary output of ParetoPrep given a start node $s$ and a target node $t$. The indicated
paths $\{s,a,b,c,t\}$ and $\{s,a,d,t\}$ are the shortest paths for the first and second criterion.
The vectors next to each node are the computed lower bound costs $\lbname$ of reaching $t$ from
the respective nodes.}
\end{figure}


\begin{figure}[t]
\centering
\begin{tabular}{ccc}
\includegraphics[width=4cm]{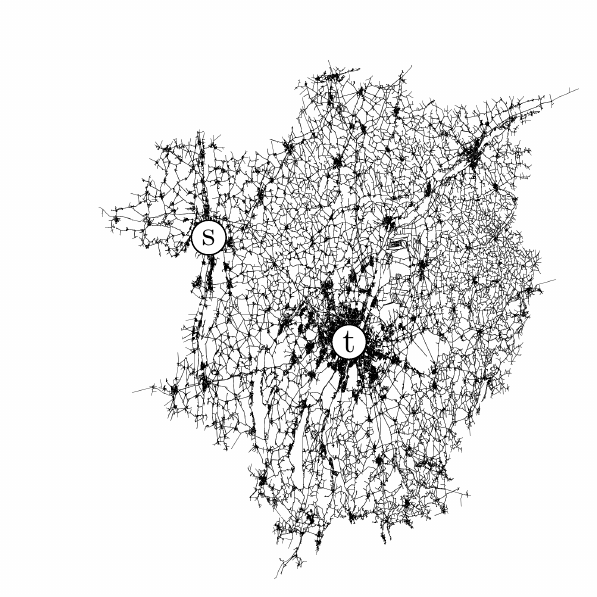} &
\includegraphics[width=4cm]{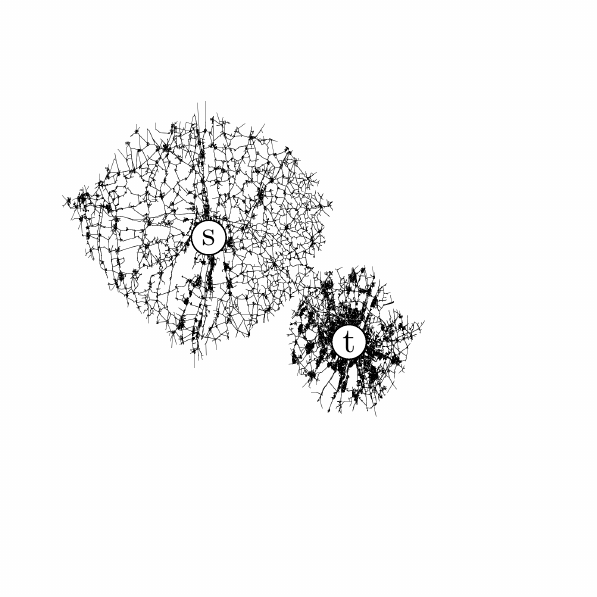} \\
Doubledijkstra & Bidirectional ParetoPrep \\
\end{tabular}
\caption{Comparison of visited nodes by Doubledijkstra and Bidirectional ParetoPrep
for a routing task from Augsburg to Munich with \te{} as cost criteria.}\label{fig:ddshortcomings}
\end{figure}

\begin{figure}[t]
\centering
\begin{tabular}{c}
\includegraphics[width=5cm]{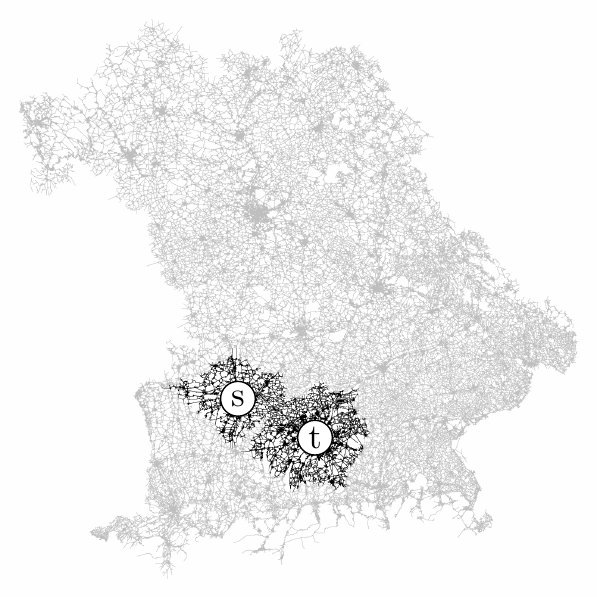}
\end{tabular}
\caption{Comparison of visited nodes by Bidirectional ParetoPrep (black subgraph) and Multidijkstra
(complete graph in gray) for a routing task from Augsburg (s) to Munich (t) with \tdecl{}.}
\label{fig:mdshortcomings}
\end{figure}

\begin{figure}[t]\label{fig:arscsearcharea}
\centering
\begin{tabular}{cc}

\includegraphics[width=4cm]{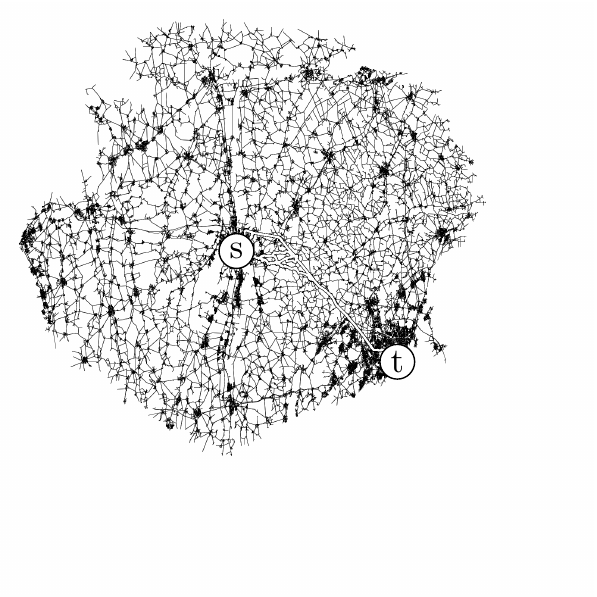} &
\includegraphics[width=4cm]{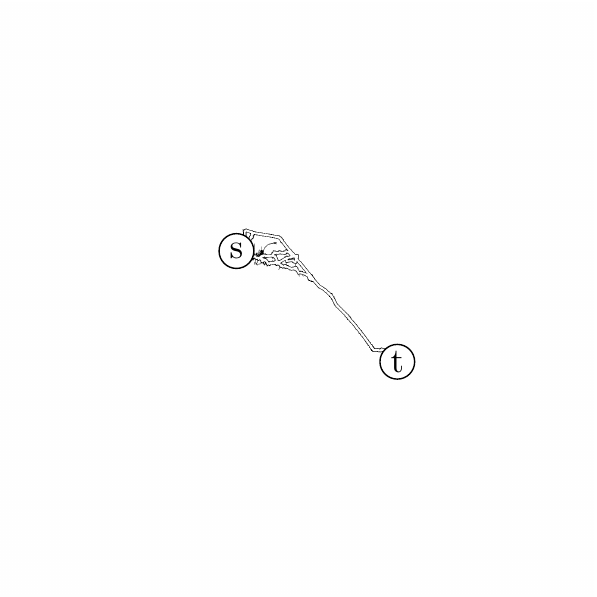} \\
without ParetoPrep & with ParetoPrep
\end{tabular}
\caption{The comparison of label correcting search areas for a routing task
from Augsburg (s) to Munich (t) with two cost criteria shows that the label correcting
search considers almost no dominated paths when using the information provided by ParetoPrep.} \label{fig:prepgraph}
\end{figure}

\subsection{ParetoPrep}

The goal of a precomputation step such as ParetoPrep or MD is to
compute a lower bound for each cost criterion and each node that is potentially visited
when computing the path skyline with a label correcting algorithm.
The cost of the shortest path w.r.t. each cost criterion represents an optimal bound. 
The shortest paths are always part of the route skyline and any additional path
being part of the skyline cannot have any cost value being smaller than the shortest path
w.r.t. the corresponding cost criterion. The MD approach performs separate
searches for each cost criterion and visits the entire network for each criterion.

The idea of ParetoPrep is to avoid traversing the same subpaths multiple times
by computing the bounds in a single traversal of the relevant portion of the graph. Furthermore, ParetoPrep
limits the search space due to the usage of a global domination check.

In this subsection, we will describe ParetoPrep and the information it computes.
In the next subsection, we will show why this resulting information allows a
fast computation of the path skyline. 
We will start our discussion by giving an overview of the used data structures.
ParetoPrep maintains a set of open nodes $open$ and a set $\paths$ of paths from
$s$ to $t$. Each visited node $n$ has an entry consisting
of two vectors $\{ \lb{n},\ succ(n)\}$.

The cost vectors $\lb{n} \in \vectors{d}$\ are the lower bound
costs of all paths from $n$ to $t$, through which $n$ was previously reached in
ParetoPrep. Upon termination of the algorithm $n$ has been reached by all
\globally{nondominated} paths from $n$ to $t$. The edges $\succ{n} \in \edges^d$
are the first edges of the currently shortest path from $n$ to $t$ w.r.t.
criterion $i$. These successor edges are used to reconstruct the shortest
paths from $s$ to $t$ w.r.t. to each of the cost criteria.

An entry of an unvisited node $n$ is assumed to be $\lb{n}_i = \infty$,
$succ_i(n) = \emptyset$.
$\lb{t}_i$ is the zero vector because the lower bound costs of reaching $t$ from $t$
are zero.

The pseudocode of the algorithm is provided in Algorithm~\ref{alg:main}. Let $s$
be the start node of the routing task, $t$ the target node and $d$ the number
of cost criteria.

\begin{algorithm}[t]


	

\SetKwInOut{Input}{input}
\SetKwInOut{Output}{output}


\DontPrintSemicolon
\BlankLine
\tcp{step 1: initialisation}
\BlankLine
$\paths \leftarrow \{\}$\;
$\open{} \leftarrow \{ t \}$ \;
\BlankLine
\tcp{step 2-5: main loop}
\BlankLine
\While{$\open{} \neq \{ \}$}{

	\BlankLine
	\tcp{step 2: node selection}
	\BlankLine
	$n \leftarrow \argmin_{q \in \open} \sum_{i = 1}^{d} \lb{q}_i$ \; 
	$\open \leftarrow open \ \setminus \ \{n\}$ \;

	\BlankLine
	\tcp{step 3: global domination}
	\BlankLine
	\If{$\dom{\paths}{ \lb{n} }$}{
		\BlankLine
		skip step 4 and 5 \;
	}
	
	\BlankLine
	\tcp{step 4: node expansion}
	\BlankLine 
	\ForEach{incoming edge $(m,n)$ of $n$}{
	
		\BlankLine
	
		\ForEach{criterion $i$}{
			\BlankLine
			\If{$\lb{n}_i+\edgecost{m}{n}_i < \lb{m}_i $}{
				\BlankLine
				${\lb{m}}_i \leftarrow \lb{n}_i+\edgecost{m}{n}_i$ \;
				$\succi{m} \leftarrow (m,n)$ \;
				$\open{} \leftarrow \open{} \ \cup \ \{m\}$ \;
			}
		}
	}

	\BlankLine
	\tcp{step 5: path construction}
	\BlankLine
	
	\If{$\lb{s}$ was modified in step 4}{
	\ForEach{modified component $i$ of $\lb{s}$}{
		\BlankLine
		
		$p \leftarrow \text{reconstructpath}(s,t,\succiname{},i)$ \;
		$\paths \leftarrow \paths \ \cup \ \{p\}$ \;
		remove paths dominated by $p$ from $\paths$ \; 
	}
	}
}
\BlankLine
RETURN $\paths$ AND $lb$
\BlankLine

\caption{Pseudocode of ParetoPrep}
\label{alg:main}
\end{algorithm}


The first step of the algorithm is the initialization. The open set is created and the target
node $t$ is added to the open set. 
The second step is node selection. In each iteration, an open node $n$ is
selected and removed from the open set. To reduce the number of nodes which have to be visited twice,
the nodes closest to $t$ should be visited first. To achieve this, each
node is ranked by the linear sum over the cost vector and the node with the
smallest value is selected first. 

The third step is a check whether the selected node $n$ has to be extended. In
other words, we compare the current lower bound $\lbcost{n}{t}$ to the current
set of shortest paths $\paths$ from $s$ to $t$. If the vector $\lb{n}$ is globally dominated,
the node does not need to be extendend.
This means, there cannot be any path from $s$ to $t$ via $n$ which is not
dominated by the already found shortest paths in $\paths$.

If $\lb{n}$ is not globally dominated, we perform step four
and five. The fourth step is the extension of the selected node $n$. In this
step, we consider the neighboring nodes having a directed edge ending at $n$,
i.e. the predecessors of $n$.
The cost of each predecessor node $m$ for each cost dimension $i$ is set to
the minumum of $\lb{m}_i$ and $\lb{n}_i + \edgecost{m}{n}_i$. For each criterion
$i$ in which $c(m)_i$ is changed, the $i$-th predecessor \arc $succ_i(m)$ is set
to $(m,n)$. If $\lb{m}$ was changed and $m$ is not the start node $s$, $m$ is
added to the set of open nodes. 

\begin{algorithm}[t]
\DontPrintSemicolon
\KwData{$s$, $t$, $\succiname{}$, $i$}
\KwResult{Current shortest path from $s$ to $t$ for criterion $i$}
$m \leftarrow s$\;
$p \leftarrow \text{new empty path} $\;
\While{$m \neq t$}{
	\BlankLine
	$(m,n) \leftarrow \succi{m}$\;	
	$p \leftarrow p \text{ extended with } (m,n)$\;
	$m \leftarrow n$\;		
}
\Return $p$\;
\caption{Pseudocode of ParetoPrep's path construction routine}
\label{alg:pathconstruct}
\end{algorithm}

The fifth step is the construction of paths from $s$ to $t$. Let us note that
$\paths$ contains the set of shortest paths w.r.t. to each cost criterion and thus,
$\paths$ is a subset of the route skyline being computed in the following label
correcting search.
This path construction step is only performed if at least one component of the vector $\lb{s}$
was modified in the previous step. For each modified cost criterion the currently shortest
path from $s$ to $t$ is constructed. These paths are constructed by following the respective successors
$\succiname{}$, similarly to how paths are reconstructed in Dijkstra's search.
The pseudocode is provided in Algorithm~\ref{alg:pathconstruct}. 

If after an iteration there are no more open nodes, the algorithm terminates.
Upon termination, $\paths$ contains a shortest
path from $s$ to $t$ for each cost criterion. And for every node $n$ which could 
be part of a \nondominated{} path from $s$ to $t$, $\lb{n}$ contains the lower bound
costs of reaching $t$.

\begin{figure}[t]\label{alg:execution}
\centering
\includegraphics[width=7cm]{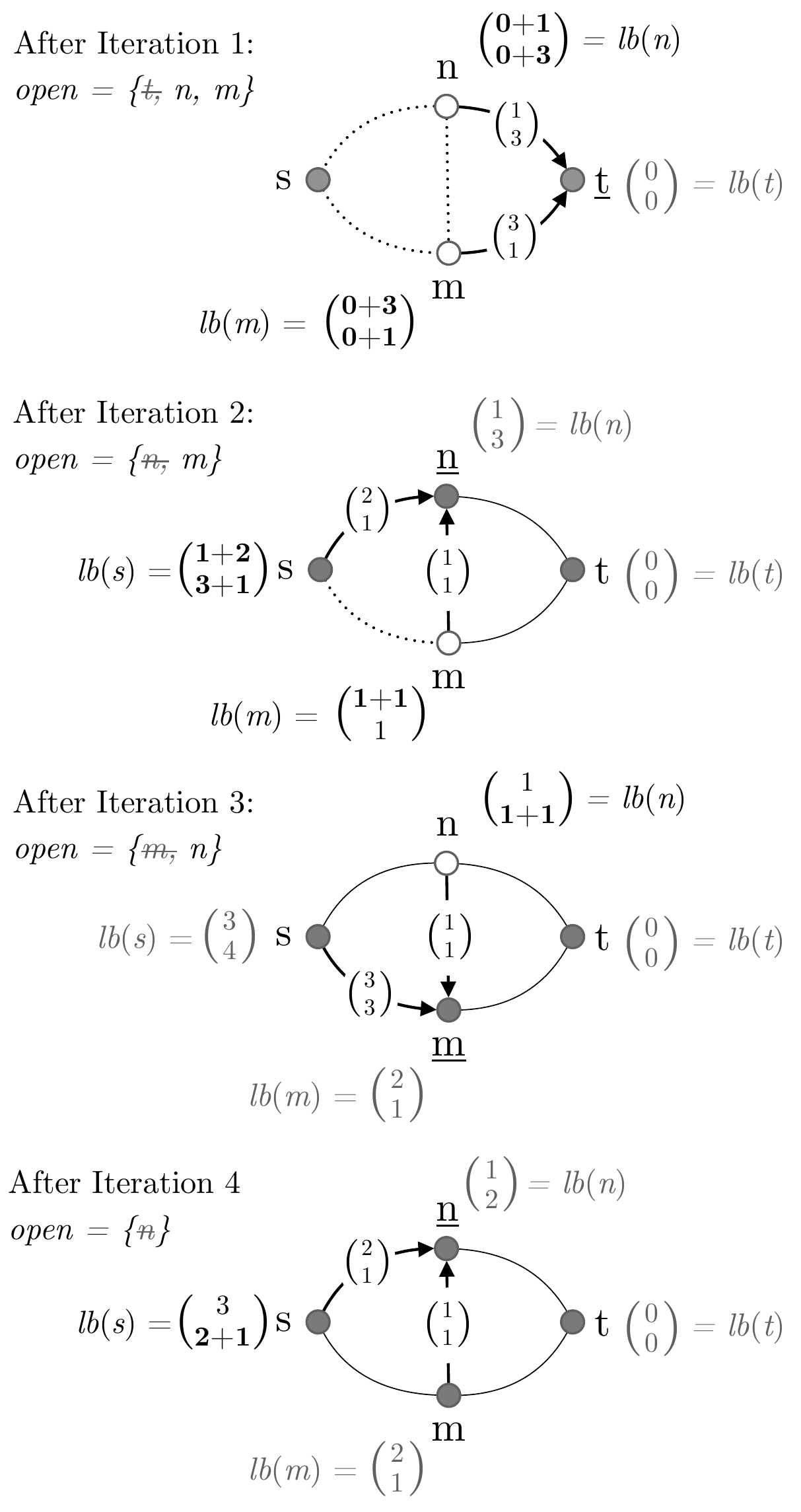}

\caption{Exemplary execution of ParetoPrep. The selected node of the iteration is underlined.
After iteration 2 the path through $[s,n_1,t]$ with the costs $[3,4]$ and after
iteration 3 the path through $[s, n_1, n_2, t]$ with the costs $[6,3]$ is constructed.
ParetoPrep terminates after iteration 4.}
\end{figure}

Figure~\ref{alg:execution} illustrates an exemplary execution of the algorithm
for a simple search task.

\subsection{Formal Aspects of ParetoPrep}
To show that ParetoPrep is a correct preprocessing step for computing path
skylines, we will show that ParetoPrep computes valid lower bounds for all nodes
that have to be visited during the search. Furthermore, we show that all single
criterion shortest paths between $s$ and $t$ are found. We will start by
introducing the relevant subgraph of $G(\vertices,\edges,\costs)$ which
is specific for a path skyline query between $s$ and $t$.
 
\begin{definition}[Solution Graph]
Let $\soledges(s,t) \subseteq \edges$
denote the set of all edges contained in some $p\in\paths(s,t)$.
Furthermore, let $\vertices_\text{SOL}:=\{v\in\vertices\mid v\in\soledges(s,t)\}$
and $\costs_\text{SOL} := \costs|_{\soledges(s,t)}$. The \emph{Solution Graph}
is defined as 
\\$\solgr(s,t) := (\vertices_\text{SOL}, \soledges(s,t), \costs_\text{SOL}).$
\end{definition}

There are two important properties of $\solgr(s,t)$ making it the central
concept for examining the correctness of ParetoPrep:
\begin{itemize}
\item It is not necessary to visit any edge $e \not \in \soledges(s,t)$ during
the search.
\item The cost vector of single criterion shortest paths in $\solgr(s,t)$
represent a viable lower bound for computing the path skyline.
\end{itemize}

The first property follows from the definition of $\solgr(s,t)$. The second
property is shown by the following lemma:

\begin{lemma}
Let $p=((s,n_1),..,(m,n),..,(n_l,t)) \in \mathcal{S}(s,t)$ be a nondominated
path, let $q=((s,n_1),..,(m,n))$ be a subpath of $p$ and let $\lb{n}_\text{SOL}$ be
the vector of the single criterion shortest paths between $n$ and $t$ for each
criterion in $\solgr(s,t)$. Then $\pathcost{q}+\lb{n}_\text{SOL}$ cannot be
dominated in $\mathcal{S}(s,t)$.
\begin{proof}
The single criterion shortest paths between $n$ and $t$ in $\solgr(s,t)$ lower
bounds any path between $n$ and $t$ in $\solgr(s,t)$. $p$ is nondominated
and for all $ 1\leq i \leq d$: $\pathcost{p}_i \leq \lb{n}_\text{SOL}+ \pathcost{q}_i$.
Thus, $\pathcost{q}+\lb{n}_\text{SOL}$ cannot be dominated in $\mathcal{S}(s,t)$.
\end{proof}
\end{lemma} 

We can now show two properties of ParetoPrep which prove the correctness
of our approach: 
\begin{itemize}
  \item ParetoPrep visits every edge in $\solgr(s,t)$.
  \item For every node $n \in \solgr(s,t): \lb{n}_i \leq \pathcost{p_i}_i$, where $p_i$ is the shortest path
  in $\solgr(s,t)$ w.r.t. cost dimension $i$.
\end{itemize}
The first property implies that every necessary node or edge is visited.
The second property implies that the lower bounds for the nodes relevant 
to the solution are correct. 
Note that for any node $n \not \in \solgr(s,t)$ $\lb{n}$ does not have
to be a correct lower bound of the cost of paths between $n$ and $t$.
An incorrect lower bound would imply larger values which could lead to
excluding $n$ from the search. However, since $n \not \in \solgr(s,t)$, it is not
required to follow paths through $n$. 
We will start by defining the graph ParetoPrep visits when traversing
$G(\vertices,\edges,\costs)$:

\begin{definition}[ParetoPrep Graph]
Let $\ppedges(s,t)\subseteq\edges$ denote the set of all edges which
are considered in Step 4 of Algorithm \ref{alg:main}. Analogously, $\vertices_\text{PP}$
denotes the set of vertices of $\ppedges(s,t)$ and $\costs_\text{PP}$ denotes the
costs restricted to $\ppedges(s,t)$. Finally, we refer to $\ppgr(s,t) := (\vertices_\text{PP}, \ppedges(s,t), \costs_\text{PP})$
as the \emph{ParetoPrep Graph}.
\end{definition}

Now, the first property is formulated and shown by the following lemma: 

\begin{lemma}
\label{lem:subset}
Given a path skyline query from $s$ to $t$ in a multicost network
$G(\vertices,\edges,\costs)$, it holds: $\solgr \subseteq \ppgr$ (for the pair $(s,t)$
which is omitted here for reasons of clarity).
\begin{proof}
It suffices to show that $\soledges \subseteq \ppedges$. Then the claim follows by
construction of the respective graphs. We show the proposed by contradiction:
Suppose, there exists $\soledges \ni (n_{i-1},n_i) \notin \ppedges$. Let
$p = ((s,n_1),\dots,(n_k,t)) \in \paths$ denote a skyline path which contains edge
$(n_{i-1},n_i)$. Since $(n_{i-1},n_i) \notin \ppedges$, the lower bound cost vector of some
${n_j, j\geq i}$, must have been globally dominated (cf. Step 3 in Algorithm \ref{alg:main}).
Since $\lb{n_j}_i \leq cost((n_j,n_{j+1}),\dots,(n_k,t))_i$ for all cost dimensions $i$,
the subpath from $n_j$ to $t$ is also globally dominated. But the negation of
Lemma \ref{lem:localdomination} states: If a subpath is dominated, so is the whole path.
Hence, $p$ must be dominated which contradicts the assumption.
\end{proof}
\end{lemma}

The second property we have to show is that the lower bounds for any node $n
\in \solgr(s,t)$ are correct:

\begin{lemma}
Given a path skyline query from $s$ to $t$ and its solution graph
$\solgr(s,t)$, let $p_i(n,t)$ denote the shortest paths between node $n
\in \solgr(s,t)$ and $t$ for cost criterion $i$ and let $\lb{n}_i$ be the
lower bound value being computed by ParetoPrep for $n$ and $i$. Then,
$\lb{n}_i \leq \pathcost{p_i(n,t)}$.
\begin{proof}
Case (i): $p_i(n,t) = ((n,n_1),..(n_l,t))$ is the shortest path w.r.t. criterion
$i$ in both graphs $\solgr(s,t,)$ and $\ppgr(s,t)$.
We now show that ParetoPrep traverses $p_i(n,t)$ and thus, $\lb{n}_i = \pathcost{p_i(n,t)}_i$.
This is proven by induction.
\\Basis: $(n_l,t)$ is traversed by ParetoPareto because any incoming node of $t$
is examined. Inductive step: Given that $((n_k,n_{k+1}),..,(n_l,t))$ is
traversed, then $\lb{n_k}_i$ is updated and $n_k$ is inserted into the open set.
Because ParetoPrep does not terminate until the open set is empty, $n_k$ will be
processed at some iteration. Now, if $\lb{n_k}$ is currently
globally nondominated in step 3, step 4 is performed and $(n_k-1,n_k)$ is
examined which means the claim is proven.
If $\lb{n_k}$ is currently globally dominated in step 3,  we have to show that
there has to be a later point in time where $n_k$ is globally nondominated.
Since $n_k \in \solgr(s,t,)$ and $\solgr(s,t) \subseteq \ppgr(s,t)$, we know
that $n_k$ is guaranteed to be nondominated at the end of ParetoPrep. Thus,
there must exist an iteration where $n_k$ passes step 3 for the first time and 
step 4 is performed traversing the final edge $(n_k-1,n_k)$.
Case (ii): $p_i(n,t)$ is the shortest path w.r.t. criterion $i$ in $\solgr{s,t}$
but not w.r.t. $\ppgr(s,t)$. In this case $\lb{n}_i < \pathcost{p_i(n,t)}_i$
holds because $\solgr(s,t,) \subseteq \ppgr(s,t)$ and the cost of any shortest path cannot
be increased by extending the graph.
\end{proof}
\end{lemma}

After proving that ParetoPrep visits all relevant nodes $n \in \solgr(s,t)$ and
computes valid lower bounds $\lb{n}$, we show that all single criterion shortest
paths between $s$ and $t$ for the complete graph are computed.
\begin{corollary}
Given a multicost network $G(\vertices,\edges,\costs)$ and a 
path skyline query from $s$ to $t$, it holds: For all cost dimensions $i$
$\lb{s}_i$ is the cost value of the shortest path from $s$ to $t$
(w.r.t. dimension $i$) in the original graph $G(\vertices,\edges,\costs)$.
\begin{proof}
The single criterion shortest paths between $s$ and $t$ are part of
$\mathcal{S}(s,t)$ and also shortest paths in the complete graph
$G(\vertices,\edges,\costs)$. Thus, case (i) of the previous lemma can be applied.
\end{proof}
\end{corollary}

Hence, ParetoPrep computes all single criterion shortest paths between $s$
and $t$ (w.r.t. $G(\vertices,\edges,\costs)$). Also, ParetoPrep visits the
subgraph that is relevant for answering a path skyline query and therein
computes valid and tight lower bounds during one single traversal of $\ppgr$.

\subsection{Bidirectional ParetoPrep}
Now that we have introduced the basic version of our algorithm and proven
its theoretic foundation, we will introduce an
optimized version of ParetoPrep. The idea of this optimization is the fact 
that the global domination check limiting the search
space can still be improved. This is achieved by employing bidirectional search.
In ParetoPrep we do not extend a node $n$ if $\lb{n}$ is dominated by some path in $\paths$.
Though this check is correct, it is not tight in most cases because the cost
of the path connecting $s$ and $n$ is disregarded. Thus, if we had a lower bound approximation
for the costs of getting from $s$ to $n$ $\lbcost{s}{n}$, the check for global
domination in Step 3 of Algorithm \ref{alg:main} could be optimized by checking if
$\dom{\paths}{\lb{n}+\lbcost{s}{n}}$.

The idea of bidirectional ParetoPrep is to start two searches simultaneously.
A backward search starts at $t$ and proceeds in the same way as described
above to derive lower bounds between the visited nodes and the target $t$.
Additionally, a forward search starts at $s$ and traverses the graph
in the forward direction in order to find lower bounds for the path between $s$
and intermediate nodes $n$.

\begin{figure}[t]
\centering
\includegraphics[height=4cm]{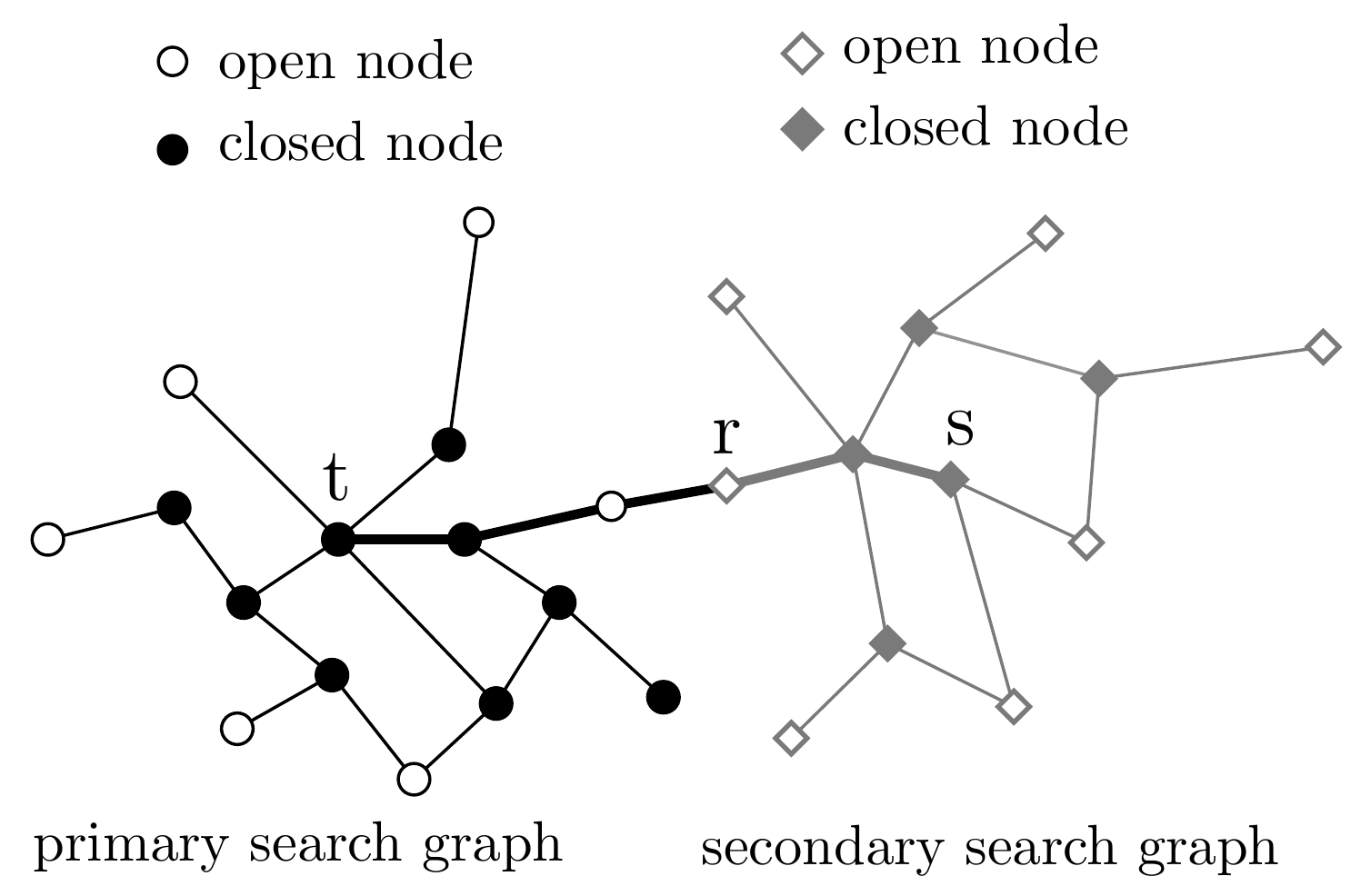}
\caption{Bidirectional ParetoPrep, when both its ParetoPrep searches rendezvous at node $r$.
To reach $s$ from outside the secondary search graph costs at least as much as the minimum of all
$\lbbetaname$ cost vectors of the secondary search.} \label{fig:bidirec}
\end{figure}

At the beginning both searches take turns, i.e. iterations are alternated.
Once both searches meet at some rendezvous node $r$, as shown in
Figure~\ref{fig:bidirec}, a path from $s$ to $t$ can be constructed for each cost dimension $i$ by merging the
current shortest paths from $s$ to $r$ and $r$ to $t$. After that, only the
backward search is continued.

In the following, we will refer to paths ending at nodes in the open set of the
forward search as open paths. If $t$ was not reached by a nondominated path $p$
from $s$ to $t$, there has to be an open path, which is contained in $p$. If
there is no such open path, ParetoPrep cannot reach $s$ through $p$.

Let $mincost^{forward}$ be the minimum cost of all open paths of the
forward search. The cost vector $mincost^{forward}$ consists of the
lower bound costs of all \globally{nondominated} paths from $s$ to any
node which has not been visited by the forward search. The minimum of all
cost vectors of open nodes yields the minimum cost of all open paths:

$${mincost}^\beta_i = \min_{ o \in \open_\beta} \lbbeta{o}_i$$  

Let $\lbcost{s}{n}$ denote the lower bound costs of all paths from $s$ to $n$ contained in
a nondominated path from $s$ to $t$. If no such costs are known, they have to be assumed as
zero. If $n$ was not reached by the forward search, they can be set to
${mincost}^\beta$:
$$\lbcost{s}{n}_i \ge \begin{cases} 0, 					& \text{if }n\text{ not visited by }\beta\text{ search}\\
	{mincost}^\beta_i,	& \text{if }n\text{ visited by }\beta\text{ search}
\end{cases}$$

The cost vector ${mincost}^{forward}$ is computed once when both searches rendezvous.
The forward search is discontinued after the rendezvous. The cost vectors
$\lbcost{s}{n}$ are then used in the backward search to determine if a node was
only reached by \globally{dominated} paths. If that is the case,
the node can be pruned.

In conclusion, this bidirectional ParetoPrep yields greater
pruning power and thus, it restricts the search space even further.
\section{Experiments}\label{section:experiments}

\vspace{1.5eM}
\textbf{Experimental Setup:} The experiments we present in this section aim to show three things. Firstly, any label correcting search (LCS) for all \nondominated{} paths
between designated start and target nodes in a network should utilize global domination. Secondly, our proposed
algorithms ParetoPrep (PP) and Bidirectional ParetoPrep (BPP) outperform the state of the art approaches Doubledijkstra (DD) and 
Multidijkstra (MD) in terms of runtime as well as memory usage. Thirdly and finally, every LCS can be significantly accelerated
in every way by employing the proposed preparatory search of ParetoPrep (PP) and even further accelerated by employing BPP. This
also allows solving more complex tasks (e.g. multicost path skyline computation) in less time.

The LCS we implemented as comparative method is ARSC \cite{arsc} (cf. Algorithm \ref{alg:arsc}) but any other LCS yields similar results. The
Reference Node Embedding, a preprocessing step which is presented in \cite{arsc}, is not used unless stated otherwise. 
Any lower bound costs are either set to zero or provided by a prior search like ParetoPrep. Note that any preprocessing step which is not
executed at query time (like the Reference Node Embedding) cannot be applied to dynamic graphs. As mentioned before (cf. Section \ref{sec:twocriteria}),
for the special case of two cost criteria an optimized variant of ARSC is used which maintains skylines in sorted lists and utilizes
the sort order when merging skylines \cite{brumbaugh,skriver2000}. 

The experiments were conducted on the road network of the German state of Bavaria, provided by OpenStreet Maps and consisting of 
1\,023\,561 nodes and 2\,418\,437 edges. There are two sets of routing tasks. The first set is comprised of ${27 \choose 2} \cdot 2 = 702$
routing tasks with Munich's central station, Munich's airport and Munich's 25 city districts as end points. The second set is comprised
of ${10 \choose 2} \cdot 2 = 90$ routing tasks between the ten biggest cities in Bavaria.

In the experiments three variables are measured: the number of visited nodes, the number of the assembled paths and the runtime
of the search task. The number of visited nodes by DD / MD / PP / BPP corresponds to the number of node entries and thereby reflects memory usage.
The number of assembled paths by the LCS is also strongly correlated with memory usage since all paths not dominated by
assembled paths have to be stored. The runtime is measured by letting all compared algorithms perform each task three times and taking the average. 
If a run takes longer than five minutes, it is aborted and counted as a time out.
All experiments are conducted on a machine with an Intel Xeon E5-2609 (10MB Cache, 2.40 GHz, 1066 MHz FSB) and 32 GB RAM, running SUSE
Linux 3.20.1 and Java OpenJDK IcedTea 1.7.0\_09. 

The cost criteria utilized in the experiments are travel duration (\criteria{\time}), route length (\criteria{\distance}), number of crossings
(\criteria{\crossings}), penalized travel duration (\criteria{\timelights}) and energy loss (\criteria{\energy}). Criterion
(\criteria{\time}) assumes travel speeds to equal the speed limits whereas the penalized estimate (\criteria{\timelights})
assumes additional 30 and 15 seconds for each crossing with and without traffic lights, respectively.

The energy loss estimate (\criteria{\energy}) is an synthetic model roughly derived from typical battery
capacities of electric cars and their respective ranges. It incorporates altitude differences in the following sense:
ascent increases the energy consumption by the gained potential energy and descent reduces the energy consumption
(while negative values are corrected to zero). Let us stress that the authenticity of the cost models used has absolutely no
influence on the computational benefits of our algorithms.

\begin{table}[h]
\begin{tabular}{llll}
\toprule
				& millions of 	&  & time\\
				& assembled & average & outs \\
algorithm		& paths & runtime & ($>$5min)\\
\midrule
\multicolumn{4}{c}{702 Munich tasks with \te{} (avg: 11.15 paths)} \\
\midrule
\LC 			&  0.29   		(53.3$\times$)	& 0.1s				(54$\times$) & 0 \\
\LCSS 			&  15.95   		(1.0$\times$)	& 5.44s				(1$\times$) & 0 \\
\bottomrule
\end{tabular}

\caption{Comparison of point-to-point \LC and single-source \LCSS}\label{tbl:globdom} 
\end{table}

\textbf{The Impact of Global Domination:} \Global{domination} pruning is indispensable for larger graphs and more
difficult tasks. In this experiment only a part of the graph of Bavaria (around
the city of Munich) with 221\,465 nodes and 519\,917 edges was used. This is 
because the single-source label correcting search (\LCSS{}) has to compute the
nondominated paths from $s$ to all other nodes in the graph which is not feasible for
the complete graph of Bavaria. Even though restricting the graph assists the \LCSS,
it can be observed that even in the smaller graph employing \gglobal{domination}
checks reduces the number of assembled paths by a factor greater than 53 and leads
to a reduction in runtime by the factor 54 (cf. Table~\ref{tbl:globdom}). It should be noted
that this is without utilizing lower bound costs which significantly improve
the pruning power of \gglobal{domination} checks even further.


\begin{table}[h]
\begin{tabular}{llll}
\toprule
			& millions of 					&  									& time\\
			& assembled 					& average 							& outs \\
algorithm	& paths 						& runtime 							& ($>$5min)\\
\midrule
\multicolumn{4}{c}{90 Bavaria tasks with \te{} (avg: 229.1 paths)} \\
\midrule
\textbf{\PPLC} & \textbf{5.2 	(15.2$\times$)}	& \textbf{5.8s	(6.9$\times$)}		& \textbf{0} \\
\REFLC{49} 	& 20.7  		(3.8$\times$)	& 14.4s			(2.8$\times$) 		& 0 \\
\REFLC{16} 	& 22.4   		(3.5$\times$)	& 13.5s			(2.9$\times$) 		& 0 \\
\REFLC{4} 	& 26.9   		(2.9$\times$)	& 14.9s			(2.7$\times$) 		& 0 \\
\REFLC{1} 	& $>$48.1   	(1.6$\times$)	& $>$30.3s		(1.3$\times$) 		& 2 \\
\LC 		& $>$79.6  		(1.0$\times$)	& $>$40.5s		(1.0$\times$) 		& 3 \\
\midrule
\multicolumn{4}{c}{702 Munich tasks with \tel{} (avg: 125.9 paths)} \\
\midrule
\textbf{\PPLC}	& \textbf{0.05 	(38.2$\times$)}	& \textbf{1.4s	(48.8$\times$)} 	& \textbf{0} \\
\MDLC 		& 0.05 			(38.2$\times$)	& 4.0s			(17.3$\times$)		& 0 \\
\REFLC{49} 	& $>$0.44  		(4.3$\times$)	& $>$14.4s		(4.7$\times$) 		& 10.3 \\
\REFLC{16} 	& $>$0.46   	(4.1$\times$)	& $>$16.3s		(4.3$\times$) 		& 15 \\
\LC 		& $>$1.93  		(1.0$\times$)	& $>$70.4s		(1.0$\times$) 		& 102 \\
\bottomrule
\end{tabular}

\caption{Comparison of runtime and number of assembled paths of LCS on its own and preceded by PP / MD or in
combination with a reference point embedding with $k$ evenly distributed landmarks (\REFLC{k}).}  \label{tbl:embcomp}
\end{table}

\vspace{1.5eM}
\textbf{Label Correcting Search and PP:} Utilizing lower bounds in \gglobal{domination checks} significantly reduces the number of paths
that have to be assembled. This improves runtime
and memory usage of the search. The runtime of ParetoPrep in these experiments is 
only 0.8s on average per query. This is only a small portion of the overall query time
but the reduction of the label correcting runtime is about tenfold.

Preceding a LCS with PP and MD is crucial to solve more difficult routing tasks in a manageable
time frame. As can be seen in Table~\ref{tbl:embcomp} for 702 Munich tasks with \tel, a \LC
on its own can solve less than one seventh of the Munich tasks in less than five minutes 
whereas with PP or MD all tasks are solved in less than two and a half
minutes. On average \LC is at least 17 times slower than \MDLC{} and at least 49 times slower than \PPLC{}. 

\begin{figure}[t]
\center{Munich routing tasks with \tel}
\centering
\includegraphics[scale=0.55]{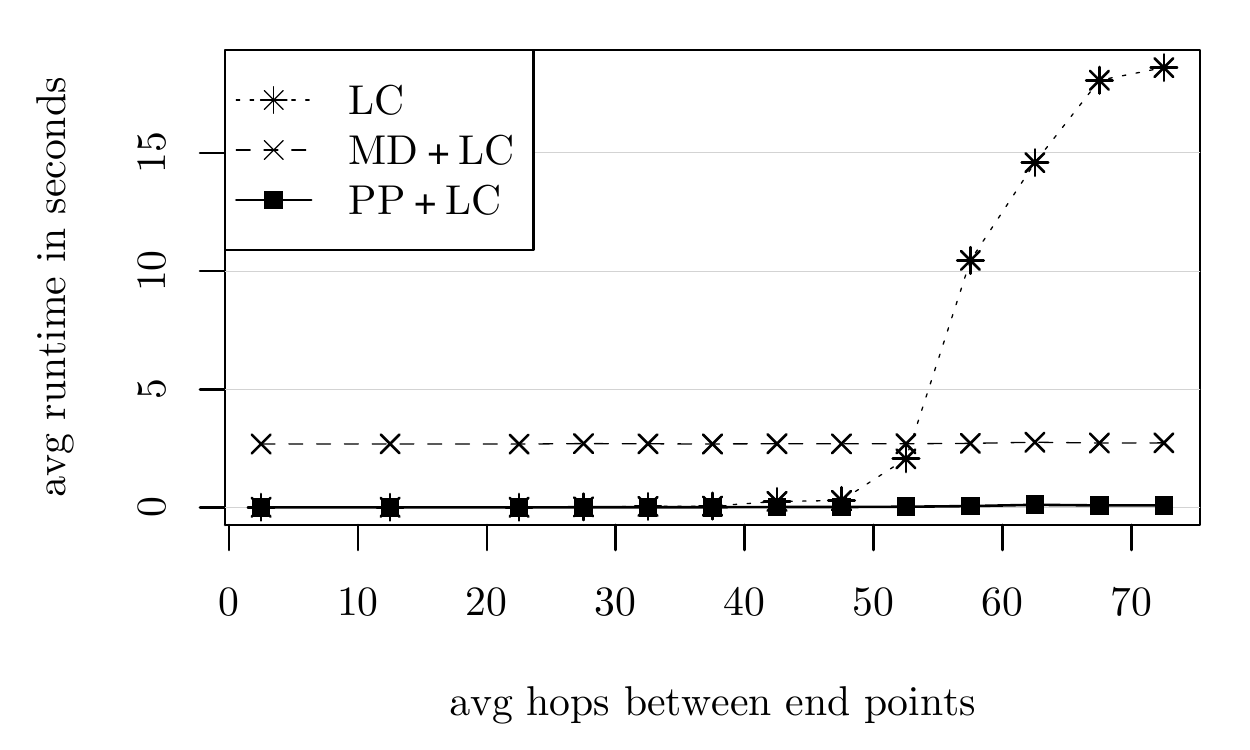}

\caption{Runtime comparison of LCS on its own and preceded by MD and PP for 290 local tasks with less than 70 hops between endpoints.}\label{fig:mdlocal}
\end{figure}

PP is clearly preferable to MD. As shown in Figure~\ref{fig:mdlocal}, \MDLC{} is often slower than
LCS on its own. \PPLC{}, on the other hand, is always faster -- on average 2.7 times faster than \MDLC{}.

Table~\ref{tbl:embcomp} also shows that a LCS in combination with PP is about twice as fast as with a
Reference Node Embedding. This result is very impressive since PP, unlike the
Reference Node Embedding, does not require any precomputation and works instantly for any given graph.
It is faster because the information provided by PP is superior to that provided by a
Reference Node Embedding and thereby reduces the number of assembled paths during the LCS.
PP also does not compute the minimum of lower bound vectors provided by
different reference nodes, which is why 49 reference nodes lead 
to worse runtime results than 16 for two cost criteria although the overall number of assembled paths is smaller.

\begin{table}[h]
\begin{tabular}{lllll}
\toprule
									& \multicolumn{4}{c}{\% of graph's nodes visited} \\
\cmidrule(r){2-5}
cost criteria							& {\scriptsize LCS} 		& 	{\scriptsize DD/MD+LC} 	& {\scriptsize PP+LC} 	& {\scriptsize BPP+LC} \\	
\midrule
\multicolumn{5}{c}{90 Bavaria tasks} \\
\midrule	
\td{} 	& 37.7 \% 	&	58 \%	& 44.9 \%	& \textbf{30.8 \%} \\
\te{} 	& 39.2 \% 	&	39.2 \%	& 44.7 \%	& \textbf{32.8 \%} \\
\midrule
\multicolumn{5}{c}{702 Munich tasks} \\
\midrule	
 \tel{} 	& 2.57 \% 	& 100 \%	& 3.87 \%	& \textbf{2.1 \%} \\
 \bottomrule
\end{tabular}
\caption{Comparison of visited nodes by an LCS on its own and when preceded by DD / MD, PP or BPP.} \label{tbl:nodesppvslc}
\end{table}

Preceding a LCS by BPP reduces the number of visited nodes. As shown in Table~\ref{tbl:nodesppvslc},
PP and DD visit more nodes than an LCS whereas BPP visits less. 

\begin{table}[h]
\centering
\begin{tabular}{llll} 
\toprule 
& \multicolumn{3}{c}{avg \% of nodes visited ($\times$ less than DD)}\\ 
\cmidrule(r){2-4} 
& DD & PP & BPP \\ 
\midrule 
 \multicolumn{4}{c}{\munichtasks} \\ 
 \midrule 
\td & $3.6\ (1.0\times)$ & $2.7\ (1.3\times)$ & $1.6\ (2.3\times)$\\ 
\te & $3.6\ (1.0\times)$ & $2.8\ (1.3\times)$ & $1.8\ (2.0\times)$\\ 
\tl & $5.4\ (1.0\times)$ & $3.6\ (1.5\times)$ & $1.7\ (3.1\times)$\\ 
\tc & $11\ (1.0\times)$ & $4.5\ (2.5\times)$ & $2.0\ (5.6\times)$\\ 
\midrule 
 \multicolumn{4}{c}{\bavariatasks} \\ 
 \midrule 
\td & $58\ (1.0\times)$ & $44\ (1.3\times)$ & $30\ (1.9\times)$\\ 
\te & $56\ (1.0\times)$ & $44\ (1.3\times)$ & $32\ (1.7\times)$\\ 
\tl & $45\ (1.0\times)$ & $41\ (1.1\times)$ & $26\ (1.7\times)$\\ 
\tc & $59\ (1.0\times)$ & $46\ (1.3\times)$ & $31\ (1.9\times)$\\ 
\bottomrule
\end{tabular}
\caption{\captiontext{\comparenodes}{\comparedalgorithmsbi}} \label{tbl:ddnodes}
\end{table} 

\textbf{Comparison of PP with DD:} This comparison is limited to two cost criteria, since DD cannot handle more than two cost criteria. 

As can be seen in Table~\ref{tbl:ddnodes}, PP and BPP visit less nodes than DD. For the Munich tasks
BPP visits from two up to almost six times less nodes than DD and almost half as many nodes for the
Bavaria tasks. Since there is one node entry for each visited node, BPP can be implemented noticeably
more memory efficient than DD.

\begin{table}[h]
\begin{tabular}{llll} 
\toprule 
& \multicolumn{3}{c}{avg runtime in ms ($\times$ faster than DD)}\\ 
\cmidrule(r){2-4} 
& DD & PP & BPP \\ 
\midrule 
 \multicolumn{4}{c}{\munichtasks} \\ 
 \midrule 
\td & $58\ (1.0\times)$ & $43\ (1.4\times)$ & $32\ (1.8\times)$\\ 
\te & $53\ (1.0\times)$ & $43\ (1.2\times)$ & $36\ (1.5\times)$\\ 
\tl & $83\ (1.0\times)$ & $55\ (1.5\times)$ & $34\ (2.4\times)$\\ 
\tc & $145\ (1.0\times)$ & $72\ (2.0\times)$ & $41\ (3.5\times)$\\ 
\midrule 
 \multicolumn{4}{c}{\bavariatasks} \\ 
 \midrule 
\td & $1118\ (1.0\times)$ & $804\ (1.4\times)$ & $726\ (1.5\times)$\\ 
\te & $1119\ (1.0\times)$ & $883\ (1.3\times)$ & $816\ (1.4\times)$\\ 
\tl & $887\ (1.0\times)$ & $717\ (1.2\times)$ & $573\ (1.5\times)$\\ 
\tc & $1057\ (1.0\times)$ & $945\ (1.1\times)$ & $768\ (1.4\times)$\\ 
\bottomrule
\end{tabular}
\caption{\captiontext{\compareruntime}{\comparedalgorithmsbi}} \label{tbl:ddruntime}
\end{table}

PP also outperforms DD in terms of runtime, as can be seen in Table~\ref{tbl:ddruntime}.
BPP is about twice as fast for the Munich tasks and about 1.5 times faster for the Bavaria tasks.

\begin{table}[h]
\begin{tabular}{llll} 
\toprule 
& \multicolumn{3}{c}{avg \% of nodes visited ($\times$ less than MD)}\\ 
\cmidrule(r){2-4} 
& MD & PP & BPP \\ 
\midrule 
 \multicolumn{4}{c}{\munichtasks} \\ 
 \midrule 
\tel & $100\ (1.0\times)$ & $3.9\ (25.8\times)$ & $2.2\ (46.0\times)$\\ 
\tdc & $100\ (1.0\times)$ & $4.8\ (20.7\times)$ & $2.4\ (41.9\times)$\\ 
\tce & $100\ (1.0\times)$ & $4.8\ (20.7\times)$ & $2.5\ (39.9\times)$\\ 
\tdecl & $100\ (1.0\times)$ & $4.7\ (21.5\times)$ & $2.4\ (41.6\times)$\\ 
\midrule 
 \multicolumn{4}{c}{\bavariatasks} \\ 
 \midrule 
\tel & $100\ (1.0\times)$ & $47\ (2.1\times)$ & $35\ (2.8\times)$\\ 
\tdc & $100\ (1.0\times)$ & $52\ (1.9\times)$ & $38\ (2.6\times)$\\ 
\tce & $100\ (1.0\times)$ & $52\ (1.9\times)$ & $39\ (2.5\times)$\\ 
\tdecl & $100\ (1.0\times)$ & $52\ (1.9\times)$ & $39\ (2.5\times)$\\ 
\bottomrule
\end{tabular}
\caption{\captiontext{\comparenodes}{\comparedalgorithmsmulti}} \label{tbl:mdnodes}
\end{table}

\begin{figure}[h]
\centering
\includegraphics[scale=0.5]{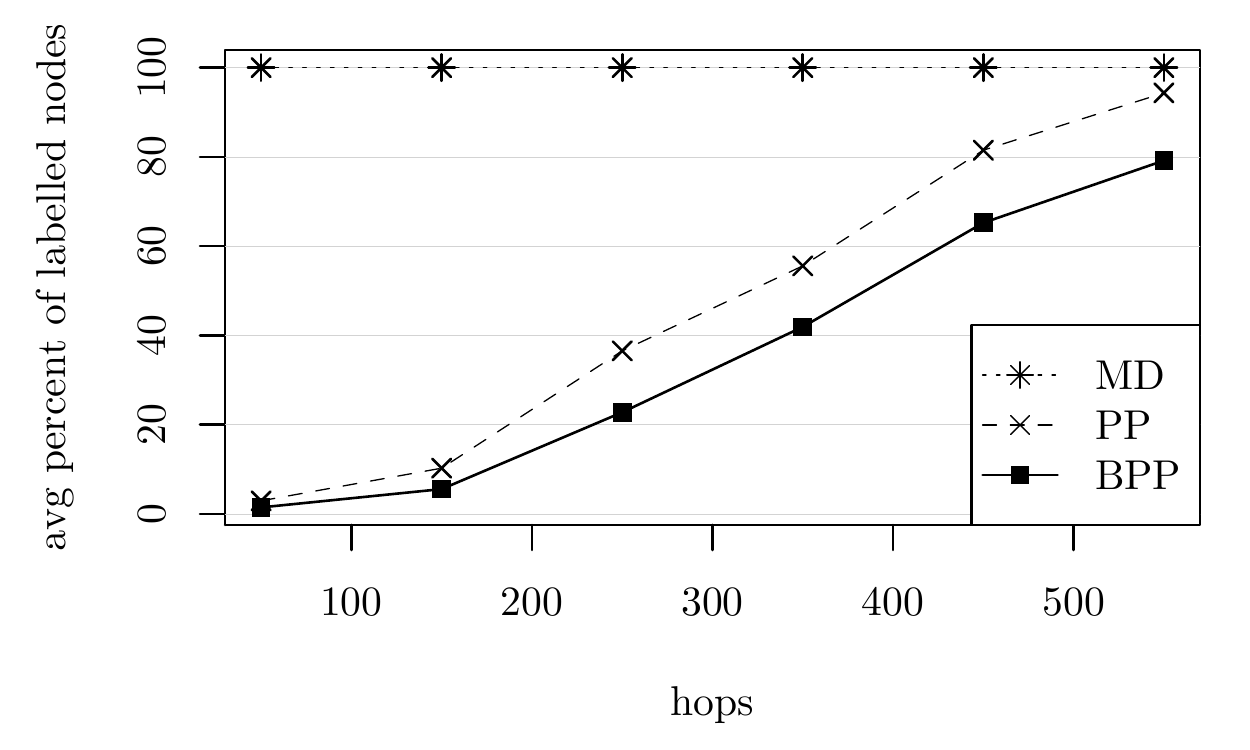}
\caption{Percentage of visited nodes for different number of hops using Munich and Bavaria Routing tasks with \tdecl{}.}
\label{fig:tdeclnodes}
\end{figure}

\vspace{1.5eM}
\textbf{Comparison of PP with MD:} Table~\ref{tbl:mdnodes} shows that PP and BPP visit only a small portion of the graph.
This is extremly important if local searches in very large graphs have to performed
which is a common use case. Even when using five cost criteria, PP and BPP explore 
a fraction of the graph (cf. Figure~\ref{fig:tdeclnodes}). In contrast, MD always visits
the complete graph once for each cost criterion. Figure~\ref{fig:mdshortcomings} illustrates
the visited area of the graph for one of these tasks and implies that
BPP would also work well on gigantic graphs, e.g. consisting of all road networks in the world.

\begin{table}[h]
\begin{tabular}{llll} 
\toprule 
& \multicolumn{3}{c}{avg runtime in ms ($\times$ faster than MD)}\\ 
\cmidrule(r){2-4} 
& MD & PP & BPP \\ 
\midrule 
 \multicolumn{4}{c}{\munichtasks} \\ 
 \midrule 
\tel & $3082\ (1.0\times)$ & $88\ (35.0\times)$ & $58\ (52.4\times)$\\ 
\tdc & $3374\ (1.0\times)$ & $94\ (35.5\times)$ & $58\ (57.7\times)$\\ 
\tce & $3204\ (1.0\times)$ & $102\ (31.2\times)$ & $65\ (49.3\times)$\\ 
\tdecl & $5323\ (1.0\times)$ & $112\ (47.1\times)$ & $70\ (75.6\times)$\\ 
\midrule 
 \multicolumn{4}{c}{\bavariatasks} \\ 
 \midrule 
\tel & $3345\ (1.0\times)$ & $1272\ (2.6\times)$ & $1163\ (2.9\times)$\\ 
\tdc & $3299\ (1.0\times)$ & $1561\ (2.1\times)$ & $1301\ (2.5\times)$\\ 
\tce & $3344\ (1.0\times)$ & $1521\ (2.2\times)$ & $1382\ (2.4\times)$\\ 
\tdecl & $5436\ (1.0\times)$ & $1690\ (3.2\times)$ & $1533\ (3.5\times)$\\ 
\bottomrule
\end{tabular}
\caption{\captiontext{\compareruntime}{\comparedalgorithmsmulti}} \label{tbl:mdruntime}
\end{table}

\begin{figure}[h]
\centering
\includegraphics[scale=0.5]{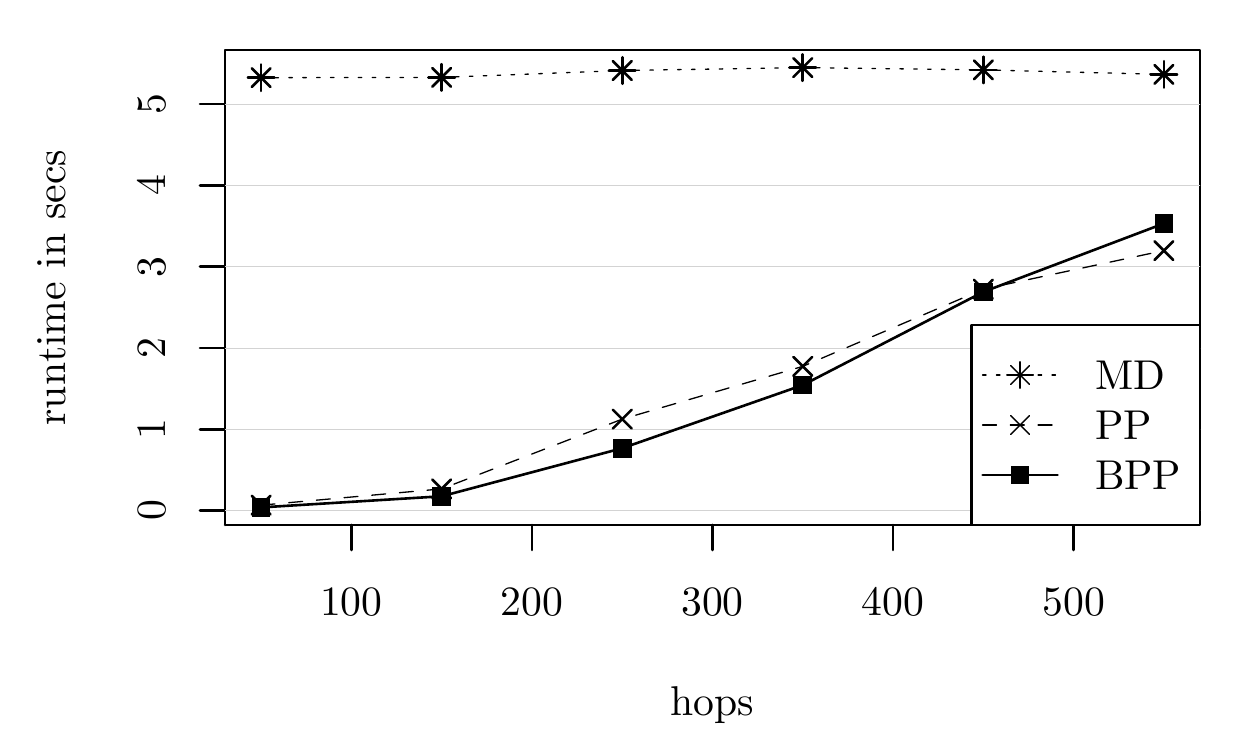}
\caption{Runtime for different number of hops using Munich and Bavaria Routing tasks with \tdecl{}.} \label{fig:tdeclruntime}
\end{figure}

Table~\ref{tbl:mdruntime} shows that PP as well as BPP are much faster than MD.
This of course is due to the limited search space. Using five criteria, BPP is
75.6 times faster than MD for the Munich tasks and 3.5 times faster for the
Bavaria tasks. As can be seen in Figure~\ref{fig:tdeclruntime} PP has the biggest
runtime advantage for local search tasks.

\section{Conclusion}\label{section:summary}
In this paper, we introduce a new efficient algorithm to compute path skyline
queries between two nodes in large multicriteria networks. Our methods
works with multiple dimensions and are suitable for dynamic networks because they do not require any precomputed
information. The core of the proposed is the algorithm ParetoPrep which
efficiently computes all single criterion shortest paths between the two query
nodes and additionally provides lower bound costs for the subgraph being relevant to the query.
The information being computed by ParetoPrep can then be employed by a label
correcting search algorithm like ARSC in two ways. First of all the set of
single criterion shortest paths allows pruning by global domination
from the start. Furthermore, the computed lower bounds provide tight forward
estimations and thus, strongly help to reduce the number of constructed paths.
We show that the results of ParetoPrep fulfill the formal requirements for this
usage. Additionally, we introduce an optimized version of ParetoPrep,
Bidirectional ParetoPrep that further limits the visited part of the graph and
hence, increases the runtime efficiency of path skyline computation. In our
experimental evaluation, we show that global domination is a key technique to
allow fast path skyline computation. Furthermore, we demonstrate that ParetoPrep
considerably outperforms other methods which compute lower bounds like Reference
Node Embedding and the state of the art multiple Dijkstra searches.

For future work, we investigate more complex versions of multicost networks
having uncertain and time-dependent cost values. Furthermore, we will investigate 
the combination of ParetoPrep with hierarchical networks for long distance
skyline computation.
\balance
\bibliographystyle{abbrv}
\bibliography{abbreviations,publication,literature}

\end{document}